\documentclass[pdflatex,iicol,sn-mathphys-ay]{sn-jnl}

\usepackage{graphicx}%
\usepackage{multirow}%
\usepackage{amsmath,amssymb,amsfonts}%
\usepackage{amsthm}%
\usepackage{mathrsfs}%
\usepackage[title]{appendix}%
\usepackage{xcolor}%
\usepackage{textcomp}%
\usepackage{manyfoot}%
\usepackage{booktabs}%
\usepackage{algorithm}%
\usepackage{algorithmicx}%
\usepackage{algpseudocode}%
\usepackage{listings}%
\usepackage{float}
\usepackage[nolist]{acronym}

\theoremstyle{thmstyleone}%
\newtheorem{theorem}{Theorem}
\newtheorem{proposition}[theorem]{Proposition}%
\newtheorem{corollary}[theorem]{Corollary}

\theoremstyle{thmstyletwo}%

\theoremstyle{thmstylethree}%

\begin{document}

\begin{acronym}
    \acro{CTMC}{Continuous-Time Markov Chain}
    \acro{SIR}{Susceptible-Infected-Recovered}
    \acro{MC}{Monte Carlo}
    \acro{RRMSE}{Relative Root Mean Squared Error}
    \acro{LHS}{Latin Hypercube Sampling}
    \acro{RFE}{Relative Frobenius Error}
\end{acronym}

\title{Learning Moment Maps for Continuous-Time Markov Chains under Monte Carlo Noise}
\author*[1]{\fnm{Madison} \sur{Pratt}}\email{mpratt10@vols.utk.edu, ORCID: 0009-0001-6975-8739}
\author[1]{\fnm{Olivia} \sur{Prosper Feldman}}\email{ofeldman@utk.edu, ORCID: 0000-0002-8042-0616}
\affil*[1]{\orgdiv{Department of Mathematics}, \orgname{University of Tennessee}, \orgaddress{\street{1403 Circle Drive}, \city{Knoxville}, \postcode{37996}, \state{TN}, \country{USA}}}

\abstract{Continuous-time Markov Chains are widely used to model stochastic dynamical systems, but key summary quantities such as means and covariances are often intractable. While Monte Carlo sampling provides asymptotically exact estimates, it becomes computationally prohibitive when moments must be evaluated across many parameter values. We develop a simulation-based surrogate modeling framework that learns parameter-to-moment mappings from Monte Carlo-derived, noise-corrupted training targets, enabling efficient and accurate approximation across the parameter space. We show that Monte Carlo noise affects mean estimation primarily through additive variance, whereas covariance estimation is additionally impacted by bias arising from nonlinear transformations of empirical estimates. Using a stochastic Susceptible-Infected-Recovered model, we demonstrate that neural networks accurately learn both mean and covariance under fixed simulation budgets allocated to constructing the noisy training labels. We further characterize how to allocate computational resources between parameter-space coverage and Monte Carlo replication, showing that covariance estimation requires a balanced allocation to control both variance and bias, while mean estimation benefits more from increased parameter space coverage. Finally, we show that the learned moment mappings produce valid population-level quantities and perform well in downstream tasks such as whitening. These results highlight the importance of accounting for Monte Carlo noise in surrogate modeling and provide practical guidance for simulation-based learning in stochastic systems.}




\maketitle

\section{Introduction}\label{intro}

\acp{CTMC} provide a natural framework for modeling stochastic dynamical systems. In these models, the system evolves over a discrete state space in continuous time, with transitions between states occurring after random waiting times, and with future evolutions depending only on the current state. The dynamics are specified through transition rates, which determine both the waiting time until the next transition and the probability of each possible state change. These rates define the infinitesimal generator of the process. This generator in turn gives rise to the master equation, which describes the time evolution of the full probability distribution over states \citep{norris1997markov, anderson1991ctmc}. \ac{CTMC} models arise in a wide range of applications, including stochastic epidemic models such as the \ac{SIR} model, where the state tracks susceptible, infected, and recovered individuals evolving through transmission and recovery events \citep{allen2008epidemic, bailey1975infectious}, as well as chemical reaction networks in which molecule counts change through reaction events \citep{anderson2015biochemical}.

In many applications, the primary quantities of interest for summarizing the dynamics of a \ac{CTMC} are low-order moments such as the mean, variance, and covariance. Although these can, in principle, be derived from the master equation, the resulting moment equations form an infinite hierarchy in which lower-order moments depend on higher-order ones, preventing closed-form solutions in general \citep{anderson2015biochemical,norris1997markov}.

One approach to addressing this challenge is to apply moment closure techniques, which truncate the hierarchy and approximate higher-order moments using assumptions about the underlying distribution. A wide range of closure methods have been proposed, including classical distribution-based closures, maximum entropy approaches, and more recent control-based formulations \citep{schnoerr2017approximation, smadbeck2013closure, wagner2023control}. Despite these advances, moment closure methods remain approximate and can introduce bias or even produce non-physical behavior in strongly nonlinear or highly stochastic regimes \citep{schnoerr2017approximation}.

An alternative approach is to estimate the moments empirically using \ac{MC} simulation. In this setting, sample trajectories of the \ac{CTMC} are generated using stochastic simulation algorithms such as the Gillespie algorithm \citep{gillespie1977exact}, and moments are approximated by averaging across realizations. \ac{MC} methods are asymptotically exact, in the sense that empirical estimates converge to the true moments as the number of simulated trajectories increases \citep{casella2002statistical, vandervaart1998asymptotic}. However, when only a finite number of simulations are available, the resulting estimates can exhibit high variance and may fail to reflect the true population-level moments.

This burden becomes particularly severe when reliable moment estimates are required across a range of parameter values, as the cost of performing many \ac{MC} simulations at each parameter scales poorly with the dimension and resolution of the parameter space. Under a fixed computational budget, this leads to an inherent trade-off between exploring the parameter space densely and obtaining accurate moment estimates at each parameter value. Allocating more simulations per parameter improves statistical accuracy but limits coverage, while increasing parameter coverage reduces the number of simulations available per parameter and increases estimation noise. While methods such as multilevel \ac{MC} \citep{giles2015mlmc, anderson2011mlmc} can reduce the cost of simulation at individual parameter values, they still rely on generating large numbers of stochastic trajectories across the parameter space, and therefore do not eliminate this fundamental trade-off.

Our goal is to mitigate this trade-off by learning parameter-to-moment mappings with surrogate models under fixed computational resources used to generate the training data, where computational cost is measured in total \ac{CTMC} simulations. Surrogate models have been widely used to approximate parameter-to-output relationships in complex systems. For example, physics-informed neural networks learn mappings from parameters to solutions of differential equations \citep{raissi2019pinn, raissi2017pinn2}, while Gaussian process–based emulators model parameter-to-output relationships and provide predictive uncertainty through posterior variance \citep{sacks1989design}.

However, a key aspect of our setting is that the parameter-to-moment mapping is learned from sample moment targets generated using a finite number of \ac{CTMC} simulations at each parameter value. Consequently, the training targets are not direct evaluations of an underlying function but are noisy approximations whose variability is induced by the estimation procedure itself. This noise is not fixed and can be reduced by the user by increasing the number of simulations per parameter, although doing so incurs additional computational cost. The level of noise in the training data is therefore directly determined by how computational resources are allocated across the parameter space.

While prior work has considered learning mappings from inputs to noisy outputs, for example through heteroskedastic Gaussian process models for stochastic simulators \citep{binois2018hetgp, BinoisGramacy2021}, these approaches assume the existence of a fixed underlying function and model noise as an additive perturbation around it. In contrast, in our setting the noise is not an intrinsic property of the system but arises from the data construction process and can be actively controlled through the simulation budget. This leads to a fundamentally different learning problem in which one must balance the accuracy of the training targets against coverage of the parameter space to optimize moment learning at a fixed computational budget.

To summarize, in this work, we develop a simulation-based surrogate modeling framework for learning parameter-to-moment mappings for \ac{CTMC} models using training targets constructed from noisy Monte Carlo estimates. We focus on learning both the mean and covariance, and ensure that the resulting covariance estimates are valid at the population level by enforcing appropriate structural constraints. A key aspect of our setting is that Monte Carlo noise affects the learning problem differently for different moments given our choice of networks. For the mean, sampling variability behaves as additive noise around the true mapping, while for the covariance, nonlinear transformations used to enforce structural constraints introduce both variance and bias in the learned targets.

We study how training data should be constructed under a fixed computational budget, showing that error depends critically on how simulations are allocated across the parameter space. In particular, we characterize the trade-off between parameter coverage and \ac{MC} replication, and demonstrate this trade-off empirically in a stochastic SIR model, leading to distinct allocation strategies for effective learning of the mean and covariance. Using this model as a representative nonlinear \ac{CTMC} system, we show that the learned surrogate models accurately recover both mean and covariance across the parameter space under fixed computational budgets. We further demonstrate that the surrogate approach enables amortization of simulation cost by allowing efficient evaluation over dense parameter regions after training, and that the learned moments can be used effectively in downstream tasks, such as whitening, indicating that the mean and covariance are jointly consistent.

The remainder of this paper is organized as follows. First, we describe the construction of the training data using \ac{CTMC} simulation and the nonlinear transformations used to obtain valid covariance matrices. Next, we introduce the neural network frameworks and discuss the effect of \ac{MC} noise in the training labels on the resulting learning targets. We then define the stochastic \ac{SIR} experiment and experimental setup used throughout the paper. Finally, we present results demonstrating the ability of the proposed framework to accurately learn mean and covariance mappings under fixed computational budgets, analyze the trade-offs between parameter-space coverage and Monte Carlo replication, and evaluate the learned moments in downstream tasks.

\section{Methods}\label{methods}

\subsection{Problem formulation}

We consider stochastic dynamical systems modeled as \acp{CTMC}, parameterized by $\theta \in \Theta \subseteq \mathbb{R}^p$. For each $\theta$, the system defines a stochastic process $\{X(t;\theta) : t \geq 0\}$ taking values in a state space $\mathcal{X} \subseteq \mathbb{R}^d$. We assume a fixed initial condition $X(0;\theta) = x_0$, and write $X(t;\theta) = \bigl(X_1(t;\theta), \dots, X_d(t;\theta)\bigr)$.

We focus on a single component of the process. For a fixed index $r \in \{1,\dots,d\}$, we define $Y(t;\theta) = X_r(t;\theta)$ and consider its evaluation at a fixed collection of discrete time points, yielding a finite-dimensional random vector $\mathbf{Y}(\theta)$.

Our objective is to approximate the parameter-to-moment mappings of the induced distribution of $\mathbf{Y}(\theta)$. For each parameter value $\theta$, we define the population mean and covariance
\[
\mu(\theta) = \mathbb{E}[\mathbf{Y}(\theta)], \qquad \Sigma(\theta) = \mathrm{Cov}(\mathbf{Y}(\theta)).
\]
These quantities are not available in closed form and are estimated using \ac{MC} simulation.
\begin{figure*}[ht!]
    \centering
    \includegraphics[width=\textwidth]{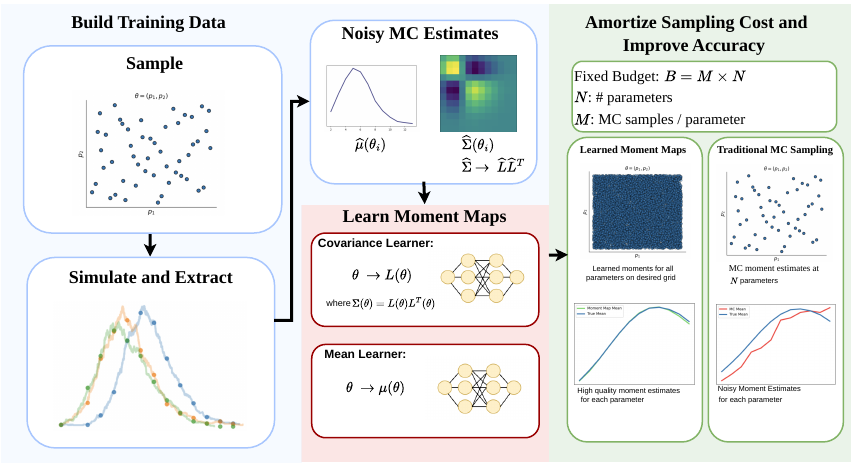}
    \caption{Schematic of the proposed parameter-to-moment mapping pipeline (Section~\ref{methods}). The left (blue) panel illustrates the generation of noisy \ac{MC} moment estimates via simulation. The middle (red) panel shows the learned parameter-to-moment mappings using neural surrogates for the mean and covariance. The right (green) panel demonstrates improved accuracy at a fixed computational budget compared to traditional \ac{MC} sampling.}
    \label{fig:Moment_Map}
\end{figure*}
Our approach proceeds in two stages. First, we sample parameter values $\{\theta_i\}_{i=1}^N$ from the parameter space $\Theta$. For each $\theta_i$, we generate \ac{MC} trajectories of the \ac{CTMC} and compute empirical estimates of the mean and covariance. Second, we use these parameter--moment pairs to train neural network surrogates that approximate the mappings $\theta \mapsto \mu(\theta)$ and $\theta \mapsto \Sigma(\theta)$. The resulting models enable evaluation of moments at new parameter values. Figure~\ref{fig:Moment_Map} provides a visualization of the overall procedure. We now describe each component of this pipeline in detail.

\subsection{Simulation-based estimation of moments}
\label{methods: training_construction}
\subsubsection{Parameter sampling}

To obtain representative coverage of the parameter space, we sample $N$ parameter values $\{\theta_i\}_{i=1}^N$ using \ac{LHS} \citep{mckay1979comparison} over $\Theta$. Under this scheme, each parameter dimension is partitioned into $N$ equally sized intervals under a uniform distribution, and one sample is drawn from each interval. The samples across dimensions are then randomly paired to form $N$ parameter vectors $\theta_i$. This construction ensures that the projections of the sample onto each coordinate axis are uniformly distributed, providing space-filling coverage of the parameter domain without clustering.

\subsubsection{CTMC simulation}

For each sampled parameter value $\theta_i$, $i = 1, \dots, N$, we generate $M$ independent sample paths of the process $X(\cdot;\theta_i)$, all initialized from the same initial condition, denoted by $\{X^{(m)}(\cdot;\theta_i)\}_{m=1}^M$. This is done using the Gillespie stochastic simulation algorithm \citep{gillespie1977exact} which is briefly described below. 

Starting from the initial condition, the system evolves through a sequence of discrete events (such as reactions or transitions), each associated with a rate that depends on the current state. These rates quantify how frequently each possible event is expected to occur per unit time. At any given state, the rates of all possible events are summed to obtain the total rate at which the system changes.

The waiting time until the next event is drawn from an exponential distribution with this total rate, so that higher overall rates correspond to shorter waiting times. Once the time of the next event is determined, a specific event is selected with probability proportional to its rate, and the system state is updated accordingly. The simulation time is then advanced by the sampled waiting time, and the procedure is repeated.

This produces trajectories that evolve through discrete jumps at random times, yielding piecewise constant sample paths consistent with the stochastic dynamics of the \ac{CTMC}. In particular, for each realization $X^{(m)}(\cdot;\theta_i)$ where $m = 1, \dots, M$, there exists a sequence of random jump times $\{\tau_k^{(m)}\}_{k \geq 1}$ such that $0 < \tau_1^{(m)} < \tau_2^{(m)} < \dots$ and $\tau_k^{(m)} \to \infty$. Between successive jump times, the process remains constant, so that for all $t \in [\tau_k^{(m)}, \tau_{k+1}^{(m)})$, $X^{(m)}(t;\theta_i) = X^{(m)}(\tau_k^{(m)};\theta_i)$.

\subsubsection{Trajectory extraction}
\label{trajectory_extraction}
In many applications, we are not interested in the full state vector, but rather in a single component of the process. For a fixed index $r \in \{1,\dots,d\}$, we define $Y^{(m)}(t;\theta_i) = X_r^{(m)}(t;\theta_i)$. Although the process evolves in continuous time, observations are typically made at discrete time points, so we consider a fixed collection of evaluation times $0 = t_0 < t_1 < \dots < t_T$.

To construct a finite-dimensional representation of each trajectory, we evaluate each sample path at these time points using left limits, defining $\mathbf{Y}^{(m)}(\theta_i) = \bigl(Y^{(m)}(t_j^{-};\theta_i)\bigr)_{j=1}^T \in \mathbb{R}^T$, where $Y^{(m)}(t_j^{-};\theta_i) = \lim_{t \nearrow t_j} Y^{(m)}(t;\theta_i)$. This choice avoids ambiguity when a jump occurs exactly at a sampling time, and each trajectory is therefore represented as a $T$-dimensional vector.

Because the initial condition is fixed, $Y^{(m)}(t_0;\theta_i)$ is deterministic and identical across all realizations. We therefore exclude $t_0$ from the representation so that $\mathbf{Y}^{(m)}(\theta_i)$ captures only the stochastic variation of the process across the $M$ simulated trajectories.

\subsubsection{Monte Carlo estimators}

Given the $M$ simulated trajectories $\{\mathbf{Y}^{(m)}(\theta_i)\}_{m=1}^M$ for a fixed parameter value $\theta_i$, we estimate the mean and covariance of the trajectory distribution using standard Monte Carlo estimators.
The sample mean is defined as
\begin{equation}
\hat{\mu}_M(\theta_i) = \frac{1}{M} \sum_{m=1}^M \mathbf{Y}^{(m)}(\theta_i),
\label{eq:mc_mean}
\end{equation}
and provides an unbiased (see Appendix~\ref{app:mean_unbiased_app} for details) estimate of the population mean
\begin{equation}
\mu(\theta_i) = \mathbb{E}[\mathbf{Y}(\theta_i)],
\label{eq:pop_mean}
\end{equation}
The sample covariance is defined as
\begin{equation}
\begin{aligned}
\hat{\Sigma}_M(\theta_i)
&= \frac{1}{M-1} \sum_{m=1}^M 
\Big( \mathbf{Y}^{(m)}(\theta_i) - \hat{\mu}_M(\theta_i) \Big) \\
&\quad \times
\Big( \mathbf{Y}^{(m)}(\theta_i) - \hat{\mu}_M(\theta_i) \Big)^\top
\end{aligned}
\label{eq:mc_cov}
\end{equation}
and provides an unbiased (see Appendix~\ref{app:cov_unbiased_app} for details) estimate of the population covariance 
\begin{equation}
\Sigma(\theta_i) = \mathrm{Cov}(\mathbf{Y}(\theta_i)),
\label{eq:pop_cov}
\end{equation}

\subsection{Learning parameter-to-moment mappings}

We train separate neural network surrogates to approximate the mappings $\theta \mapsto \mu(\theta)$ and $\theta \mapsto \Sigma(\theta)$ using the Monte Carlo estimators defined in \eqref{eq:mc_mean} and \eqref{eq:mc_cov} as noisy targets. We approximate the mappings in each case using fully connected multi-layer perceptrons (MLPs). The input to each model is the parameter vector $\theta \in \mathbb{R}^p$, and the outputs correspond to either the mean vector or a parameterization of the covariance matrix. Inputs and outputs are standardized using statistics computed from the training data.

\subsubsection{Mean mapping}

Given parameter samples $\{\theta_i\}_{i=1}^N$ and corresponding Monte Carlo estimates $\{\hat{\mu}_M(\theta_i)\}_{i=1}^N$ defined in \eqref{eq:mc_mean}, we train a neural network $f_\mu : \mathbb{R}^p \to \mathbb{R}^T$ to approximate $\mu(\theta)$ defined in (\ref{eq:pop_mean}).

We train the network by minimizing the empirical risk
\begin{equation}
\hat{R}_{N,M}(f_\mu)
=
\frac{1}{N} \sum_{i=1}^N
\| f_\mu(\theta_i) - \hat{\mu}_M(\theta_i)\|_2^2,
\label{eq:mean_loss}
\end{equation}
where $\hat{\mu}_M(\theta_i)$ is the Monte Carlo estimator defined in \eqref{eq:mc_mean}.

\subsubsection{Covariance mapping}

We next consider learning the parameter-to-covariance mapping $\theta \mapsto \Sigma(\theta)$. For each sampled parameter $\theta_i$, we construct the empirical covariance $\hat{\Sigma}_M(\theta_i)$ using \eqref{eq:mc_cov}. In exact arithmetic, this estimator is symmetric and positive semidefinite. However, in finite-precision computation, small asymmetries and loss of positive definiteness may arise due to numerical error and finite sample effects, which can prevent stable Cholesky factorization.

To address this, we apply a preprocessing pipeline consisting of symmetrization and diagonal loading. Symmetrization replaces the covariance matrix with its symmetric part, obtained by averaging it with its transpose, while diagonal loading adds a small multiple of the identity to ensure strict positive definiteness (see Appendix~\ref{app:sym} and \ref{app:loading} for details).

After preprocessing, we compute a Cholesky factorization, $\hat{\Sigma}_M(\theta_i) = \hat L_M(\theta_i) \hat L_M(\theta_i)^\top$, where $\hat L_M(\theta_i)$ denotes the Cholesky factor of $\hat\Sigma_M(\theta_i)$. This guarantees that the reconstructed covariance matrices are symmetric positive definite by construction. Let $\mathrm{vec}_{\mathrm{lt}}(\cdot)$ denote the vectorization of the lower-triangular entries of a matrix. We then construct the training target,
\begin{equation}
\hat{\ell}_M(\theta_i) = \mathrm{vec}_{\mathrm{lt}}(\hat{L}_M(\theta_i)),
\label{eq:vec_cholesky}
\end{equation}
which is the vectorized lower-triangular entries of the Cholesky factor. 

We train a neural network $f_\Sigma: \mathbb{R}^{p} \rightarrow \mathbb{R}^{T(T+1)/2}$ to approximate the true target of the vectorized entries $\ell(\theta) = \mathrm{vec}_{\mathrm{lt}}(L(\theta))$ by minimizing the empirical risk

\begin{equation}
\hat{R}_{N,M}^{\Sigma}(f_\Sigma)
=
\frac{1}{N} \sum_{i=1}^N
\| f_\Sigma(\theta_i) - \hat{\ell}_M(\theta_i) \|_2^2.
\label{eq:emp_risk_cov}
\end{equation}

At inference time the predicted Cholesky factor, $L_{\text{pred}}(\theta)$, is reconstructed from the predicted vectorized lower diagonal components. The predicted covariance is then reconstructed via $\Sigma_{\text{pred}}(\theta) =L_{\text{pred}}(\theta)L_{\text{pred}}(\theta)^\top$, and all downstream evaluations are performed on the reconstructed covariance matrices.

\section{Effect of Monte Carlo Noise on Neural Network Targets}
\label{methods: effect_on_noise}
We consider a fixed computational budget 
\begin{equation}
\label{eq: fixed_budget}
    B = MN,
\end{equation}
corresponding to the total number of simulated CTMC trajectories used to construct the training data. Under this constraint, increasing the number of trajectories per parameter $M$ reduces Monte Carlo noise in the targets, while increasing the number of parameter samples $N$ improves coverage of the parameter space.

Because the training targets are constructed from Monte Carlo estimators based on a finite number of trajectories, both learning problems are affected by sampling noise. However, this noise influences the mean and covariance mappings in fundamentally different ways. We summarize these effects below; detailed derivations are provided in Appendix~\ref{app:effect_on_mc_noise}.

\subsection{Mean mapping}

For the mean network, the targets $\hat{\mu}_M(\theta)$ defined in \eqref{eq:mc_mean} are unbiased estimators of the true mean $\mu(\theta)$. Under squared loss, training on these targets is equivalent to regression with additive noise.

At the population level, this noise contributes only an additive variance term that scales as $1/M$, without altering the underlying regression function. Consequently, the optimal predictor remains the true mapping $\theta \mapsto \mu(\theta)$ (see Appendix~\ref{app:mc_effect_mean_learner}).

Under a fixed budget $B$, increasing $M$ reduces the variance of the targets, while increasing $N$ improves the approximation of the population risk. Since the regression function itself is unaffected by Monte Carlo noise, allocating more resources toward larger $N$ may be preferable for improving generalization of the mean network.

\subsection{Covariance mapping}

The covariance learner differs fundamentally due to the nonlinear transformation applied to the Monte Carlo estimates. While the sample covariance $\hat{\Sigma}_M(\theta)$ is an unbiased estimator of $\Sigma(\theta)$, the training targets are obtained after preprocessing, Cholesky factorization, and vectorization (see \eqref{eq:vec_cholesky}).

This nonlinear transformation to construct the targets introduces bias and the resulting targets are not unbiased estimators of the population quantity $\ell(\theta) = \mathrm{vec}_{\mathrm{lt}}(L(\theta))$. As a result, the learning problem is effectively centered around a finite-$M$ regression function that depends on the Monte Carlo sample size.

Unlike the mean case, Monte Carlo noise therefore affects both the variability of the targets and the target function itself. As $M \to \infty$, this bias vanishes and the effective regression function converges to $\ell(\theta)$ under suitable regularity conditions (see Appendix~\ref{app:mc_effect_cov_learner}).

Under a fixed budget $B$, increasing $M$ improves covariance learning in two ways. First, it reduces variance in the empirical covariance estimator and second, it decreases the bias induced by the nonlinear transformation. This suggests that larger values of $M$ are more critical for the covariance learner than for the mean learner.

Overall, Monte Carlo noise acts as pure variance in the mean mapping, but introduces both variance and bias in the covariance mapping. Consequently, the two learning problems admit fundamentally different resource allocation strategies under a fixed budget $B$. These trade-offs are investigated empirically in the Results section.

\section{Experimental Setup}

We evaluate the proposed framework using the stochastic \ac{SIR} model as a representative \ac{CTMC}, which provides a challenging test case due to its nonlinear stochastic dynamics, including outbreak variability and extinction events, as well as the absence of closed-form moment equations (see Appendix~\ref{app:SIR_moment_derive}). These features make it a suitable benchmark for assessing the ability of neural networks to learn parameter-to-moment mappings from simulation data.

We focus on the infected population and aim to learn its mean and covariance over the discrete time grid \(t = 0, \dots, 14\), excluding the deterministic initial time point. The system is initialized at \((S, I, R) = (763, 3, 0)\), and all simulations are performed from this fixed initial condition. This choice of initial conditions is motivated by the famous English boarding school influenza outbreak dataset, one of the classic datasets used in infectious disease modeling to study the spread of epidemics in closed populations \citep{Anonymous1978Influenza}. The model is parameterized by the transmission rate \(\alpha\) and recovery rate \(\beta\), and we consider a biologically relevant parameter region where $\beta \in (0.00125,0.00325)$ and $\alpha \in (.1,.9)$. Under this setup, the mean corresponds to a trajectory of length 13, while the covariance is a \(13 \times 13\) matrix, making the learning of second-order dependence moderately high-dimensional and nontrivial. 

Both the mean and covariance mappings are approximated using fully connected multilayer perceptrons (MLPs) with input \(\theta \in \mathbb{R}^2\), three hidden layers of width 128, and ReLU activation functions. A dropout rate of 0.05 is applied during training to improve generalization, and inputs and outputs are standardized using statistics computed on the training data, with the same transformations applied to validation and test sets. Training is performed in Pytorch using the AdamW optimizer with learning rate \(10^{-3}\), weight decay \(10^{-5}\), and batch size 64, with early stopping based on validation loss using a patience of 50 epochs \citep{Kingma2015Adam, Paszke2019PyTorch}. Model predictions are evaluated against a test set of 1000 parameter values sampled via \ac{LHS} with high-fidelity reference moment estimates computed using \(10^5\) \ac{CTMC} simulations per parameter value.

Although mean and covariance learners are trained to minimize the empirical risk functions (\ref{eq:mean_loss}) and (\ref{eq:emp_risk_cov}) respectively, we will use additional measures in practice to see how well the predicted moments accurately reflect the population moments.

For a mean mapping model trained on $N$ parameters and $M$ replicates per parameter, the error at each individual test parameter is evaluated using the trajectory \ac{RRMSE}, defined below as

\begin{equation}
    \text{RRMSE}(\theta) =\frac{||\hat \mu(\theta)- \mu(\theta)||_2}{||\mu(\theta)||_2 + \epsilon}
    \label{eq:rel_RMSE}
\end{equation}
where $\hat \mu(\theta)$ is the predicted mean and $\mu(\theta)$ is the high fidelity reference mean.

For a covariance mapping model trained on $N$ parameters and $M$ replicates per parameter, the error at each individual test parameter is evaluated by the \ac{RFE}, defined below as

\begin{equation}
\begin{aligned}
\text{Frob}_{\text{rel}}(\theta)
&= \frac{||\hat\Sigma(\theta)-\Sigma(\theta)||_{F}}{||\Sigma(\theta)||_{F} + \epsilon}, \\
||A||_{F} &= \sqrt{\sum_{i,j} A_{i,j}^2}
\end{aligned}
\label{eq:cov_rel_F}
\end{equation}

where $\hat\Sigma(\theta)$ is the reconstructed covariance from prediction and $\Sigma(\theta)$ is the high fidelity reference covariance.

To characterize the overall performance on the trained networks over the test sets, for the mean mapping we define
\begin{equation}
    \mathrm{RRMSE}_{\text{mean}} = \frac{1}{1000} \sum_{p = 1}^{1000} \mathrm{RRMSE}(\theta_\text{test,p})
    \label{eq: avg_rel_RMSE}
\end{equation}
For the covariance mapping we define,
\begin{equation}
    \mathrm{RFE}_{\mathrm{mean}}= \frac{1}{1000} \sum_{p = 1}^{1000} \mathrm{RFE}(\theta_\text{test,p})
    \label{eq: cov_avg_rel_F}
\end{equation}
which is the average \ac{RRMSE} and \ac{RFE} over the 1000 test parameters denoted $\theta_\text{test,p}$ respectively. 

\subsection{Impact of Computational Budget Experiment Details}
\label{experiment: impact_of_budget}
One goal of this experiment is to assess how different $(M,N)$ training configurations impact overall network performance on the test set. Recall from Section~\ref{methods: effect_on_noise}, $M$ is the number of \ac{CTMC} simulations used to generate the sample estimates at a given parameter value, and $N$ is the number of different parameter locations. To this end, we train multiple mean and covariance networks across a range of $(M,N)$ configurations. For the covariance model, we construct many training data sets using the prodecure described in Section~\ref{methods: training_construction} with $N$ ranging from $25$ to $4500$ and $M$ ranging from $25$ to $5000$. All $M,N$ combinations can be viewed in Table~\ref{tab:cov_MN_config}. For the mean model, we construct many training data sets using \ref{methods: training_construction} with $N$ ranging from $100$ to $50000$ and $M$ ranging from $2$ to $200$. All $M,N$ combinations can be viewed in Table~\ref{tab:mean_MN_config}. We note that smaller values of $M$ are used for the mean model based on prior experiments and the implications of Section~\ref{methods: effect_on_noise}.

For each $(M,N)$ configuration, models are trained across multiple independent replicates. The number of training replicates was selected depending on the size of $N$. See Table~\ref{tab:mean_MN_config} and Table~\ref{tab:cov_MN_config} for all replicate choices. The number of replicates is varied across configurations to account for differences in estimator variability. In particular, configurations with smaller training budgets $B$ (Eq.~\ref{eq: fixed_budget}) exhibit higher variability due to increased stochastic noise in the training data, and therefore require a larger number of replicates to obtain stable estimates of performance. Conversely, larger budget configurations produce more stable estimates and require fewer replicates.

To summarize performance for a given $(M,N)$ configuration, we define the overall error as the average across replicates of the test-set error. For the mean model, this corresponds to the replicate-averaged $\mathrm{RRMSE}_{\mathrm{mean}}$, defined as
\begin{equation}
    \mathrm{RRMSE}_{\mathrm{rep}}(M,N) = \frac{1}{R} \sum_{k = 1}^{R} \mathrm{RRMSE}_{\mathrm{mean,k}}
    \label{eq: mean_avg_rep_avg_rel_RMSE}
\end{equation}
and for the covariance model, the replicate-averaged $\mathrm{RFE}_{\mathrm{mean}}$ defined as 
\begin{equation}
    \mathrm{RFE}_{\mathrm{rep}}(M,N) = \frac{1}{R} \sum_{k = 1}^{R} \mathrm{RFE}_{\mathrm{mean}, k}
    \label{eq: cov_avg_rep_avg_rel_F}
\end{equation}
Lastly, in addition to the main metrics listed above, we have additional metrics that are used for performance of single test parameters, and for plotting purposes defined in Appendix~\ref{app:additional_metrics}.

\section{Results}

\subsection{Neural networks accurately learn the mean and covariance}

\begin{figure*}[ht!]
    \centering
    \includegraphics[width=\textwidth]{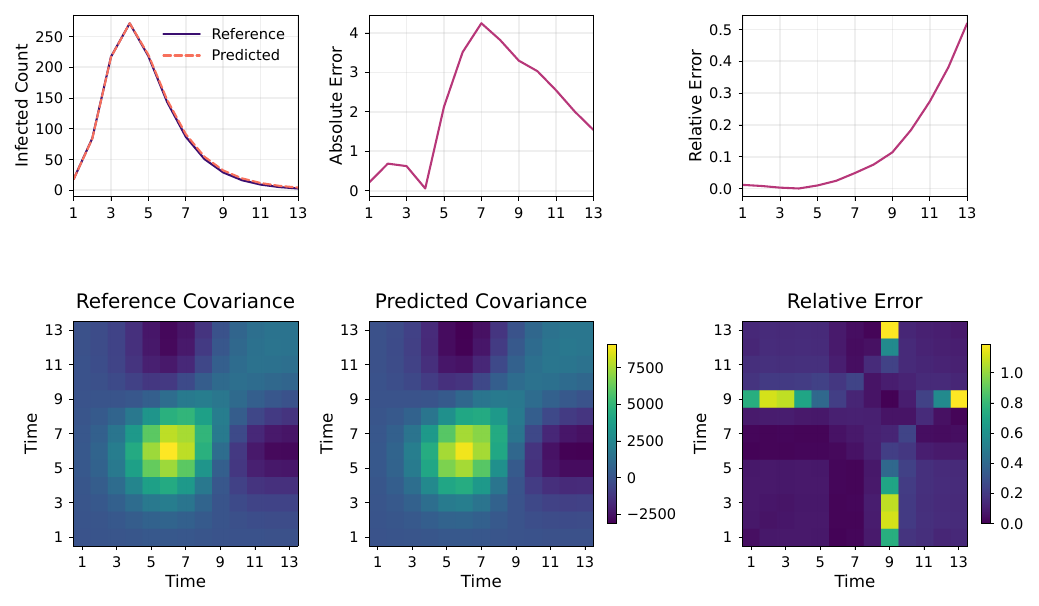}
    \caption{The top row shows the learned mean trajectory compared to the reference mean for $(\alpha,\beta \times 5\times 10^{-3}) = (0.632, 0.642)$, along with the absolute error (middle) and relative error (right) at each time point. The bottom row shows the reference covariance matrix (left) and predicted covariance matrix (middle) for $(\alpha,\beta \times 5\times 10^{-3}) = (0.138, 0.274)$, with the entry-wise relative error shown on the right.}
    \label{fig:example_mean_cov}
\end{figure*}

We first show that the neural networks can accurately learn both the mean trajectory and covariance structure under a fixed computational budget. We consider a setting where the total number of \ac{CTMC} simulations used to generate training data, $B$, is on the order of $10^5$, and train the mean and covariance networks within this budget.

For the mean model, accuracy is evaluated at each test parameter using the trajectory \ac{RRMSE} defined in equation~\ref{eq:rel_RMSE}. The median trajectory \ac{RRMSE} is $2.0\%$, with a mean of $3.2\%$ and a 95th percentile of $8.5\%$. This indicates that the network accurately recovers the mean trajectory for the vast majority of parameter settings, with only a small fraction of cases exhibiting larger errors. Notably, higher errors are localized to regions of the parameter space corresponding to large $\alpha$ and small $\beta$, where the epidemic dynamics are more variable and extinction events are more frequent. A full visualization of error across the parameter space is provided in 
Supplementary Figure 1. Overall, these results demonstrate that, under the given computational budget, the neural network reliably captures the mean behavior across the parameter space.

For the covariance model, accuracy is evaluated at each test parameter using the \ac{RFE} error defined in equation~\ref{eq:cov_rel_F}. Using the calculated test errors, the median \ac{RFE} is $7.8\%$, with a mean of $8.8\%$ and a 95th percentile of $15.2\%$. This indicates that the network is able to recover the covariance structure with reasonable accuracy across most parameter settings, although errors are generally larger than those observed for the mean model. A small subset of parameters exhibits substantially higher error, as reflected in the upper tail of the distribution. In particular, approximately $29.5\%$ of parameters have relative error exceeding $10\%$, while only $2.4\%$ exceed $20\%$, indicating that large errors are relatively rare. A visualization of the error across the parameter space, highlighting regions of higher covariance estimation error, is provided in Supplementary Figure 2. Overall, these results demonstrate that, under the given computational budget, the neural network is able to capture the covariance structure across the parameter space, albeit with greater variability than in the mean predictions.

To complement these global results, we select representative test parameters for both the mean and covariance models whose errors are closest to the median of their respective global error distributions. For the mean, the selected parameter is $(\alpha,\beta \times 5\times 10^{-3}) = (0.632, 0.642)$, with a \ac{RRMSE} of approximately $2.0\%$. For the covariance, the selected parameter is $(\alpha,\beta \times 5\times 10^{-3}) = (0.138, 0.274)$, with a \ac{RFE} of $7.8\%$. This selection ensures that the displayed results reflect typical model performance rather than a best- or worst-case scenario.

Figure~\ref{fig:example_mean_cov} (top row) shows the predicted mean trajectory alongside the reference trajectory, along with the corresponding absolute (computed using (\ref{eq: mean_entrywise_abs_error})) and relative errors (\ref{eq: mean_entrywise_rel_error})) across time. The model accurately captures the overall trajectory, with small absolute errors throughout the evaluated time points (mean absolute error $\approx 2.13$, maximum $\approx 4.25$), even near the peak of the epidemic. While the relative error appears to increase at later time points, this is primarily due to normalization by small true values as the trajectory decays toward zero, rather than a degradation in model accuracy.

The bottom row of Figure~\ref{fig:example_mean_cov} compares the predicted and reference covariance matrices, along with the entrywise relative error (computed using (\ref{eq:cov_entrywise_rel_error})). The overall covariance structure is well captured and nearly indistinguishable when looking at the true versus predicted matrices. The largest errors are localized around time point 9, corresponding to a transitional phase in the epidemic where the covariance structure shifts from high to low variability. In this regime, both the magnitude and structure of the covariance change rapidly, making the underlying relationships more difficult for the network to approximate accurately. Additionally, smaller covariance values in this region amplify relative error, further contributing to the observed discrepancies. Overall, these results demonstrate that the model accurately recovers both the mean trajectory and covariance structure for representative parameter settings.

While these results demonstrate that the neural networks can accurately learn the mean and covariance under a fixed computational budget of $10^5$, an important question remains: how does performance depend on the amount of computational resources used to generate training data?

\subsection{Increasing simulation budget reduces error}

\begin{figure*}[ht!]
    \centering
    \includegraphics[width=\textwidth]{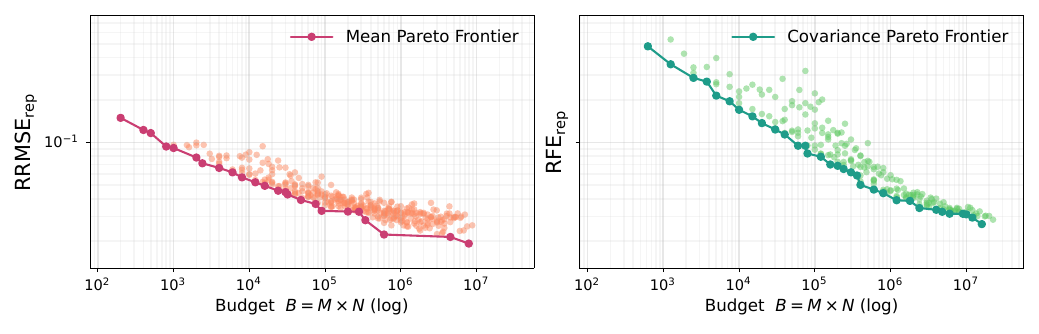}
    \caption{Each point corresponds to a choice of $M$ and $N$, where $M$ is the number of \acf{CTMC} simulations per parameter and $N$ is the number of parameter settings, with total computational budget $B$ (Eq.~\ref{eq: fixed_budget}). The left panel shows mean evaluation error (Eq.~\ref{eq: mean_avg_rep_avg_rel_RMSE}), and the right panel shows covariance evaluation error (Eq.~\ref{eq: cov_avg_rep_avg_rel_F}). The Pareto frontiers indicate the best achievable performance at each budget level.}
    \label{fig:error_vs_budget}
\end{figure*}

The dominant computational cost of our approach lies in generating training data via \ac{MC} simulation, so understanding how predictive accuracy scales with total simulation budget is critical. In this section, we study how test error depends on the total simulation budget $B$ (Eq.~\ref{eq: fixed_budget}) used to generate training data. As explained in section \ref{experiment: impact_of_budget}, this is done by training many neural networks with varying $M$,$N$ configurations and replicates, and evaluating the performance of each trained network. This allows us to explore training budgets, $B$, ranging from approximately \(10^2\) to \(10^7\) total simulations. For the mean, evaluation error at a single $M,N$ configuration is computed using the replicate-averaged $\mathrm{RRMSE}_{\mathrm{mean}}$ defined in Eq.~\ref{eq: mean_avg_rep_avg_rel_RMSE}. For the covariance, evaluation error is measured using the replicate-averaged $\mathrm{RFE}_{\mathrm{mean}}$ defined in Eq.~\ref{eq: cov_avg_rep_avg_rel_F}. 

Figure~\ref{fig:error_vs_budget} shows evaluation error versus computational budget $B$. Each point corresponds to a network trained with different configurations $(M,N)$. The Pareto frontier highlights the best achievable evaluation error at each budget.

Across several orders of magnitude in compute, we observe a clear and consistent decrease in evaluation error as the total simulation budget increases for both the mean and covariance. For the mean, the evaluation error decreases from approximately \(0.149\) at low budgets to \(0.0192\) at the highest budgets considered, corresponding to a \(7.8\times\) reduction. Similarly, the evaluation error for the covariance decreases from \(0.479\) to \(0.0264\), an \(18.1\times\) reduction. This trend is further illustrated in the 
Supplementary Material Figure 3, which shows how error across the parameter space becomes both smaller in magnitude and more uniform as the simulation budget increases.

At higher budgets, the rate of improvement slows, indicating diminishing returns. For example, beyond \(B \approx 9 \times 10^4\), the mean learner prediction error decreases more modestly from \(0.0328\) to \(0.0223\) over roughly an order of magnitude increase in compute. A similar trend is observed for the covariance, where beyond \(B \approx 8 \times 10^4\), the prediction error decreases from \(0.0836\) to \(0.0438\).

Although total computation strongly governs overall accuracy, we observe a noticeable variation in error between configurations with similar values of \(B\), as reflected by the spread of points in Figure~\ref{fig:error_vs_budget}. For example, near \(B \approx 10^5\), the mean error ranges from approximately \(0.0328\) to \(0.0573\), while the covariance error ranges from \(0.0698\) to \(0.321\). This variability indicates that total budget alone does not fully determine performance.

These results demonstrate that total simulation budget is the primary driver of predictive accuracy. However, because each budget \(B\) is achieved through different combinations of \(M\) and \(N\), the observed variation in error at fixed \(B\) suggests that evaluation error also depends on how this budget is allocated between the number of parameter settings, $N$, and the number of \ac{MC} replicates per parameter, $M$. We investigate this effect in detail in the next section.

\subsection{Allocation of simulation budget impacts performance}

\begin{figure*}[ht!]
    \centering
    \includegraphics[width=\textwidth]{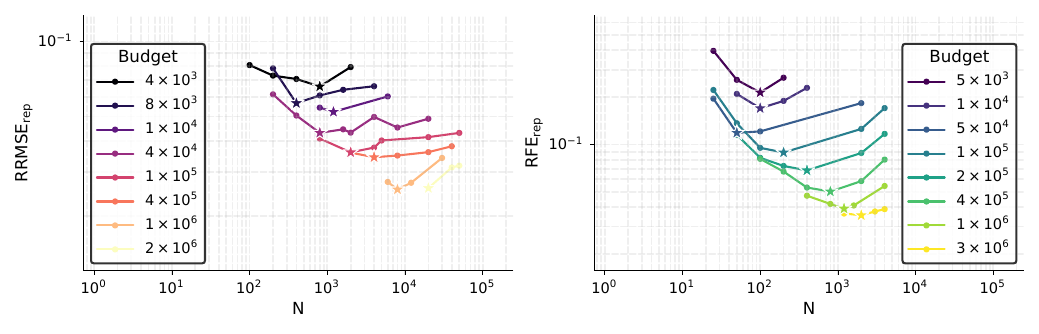}
    \caption{The left panel shows the mean network evaluation errors (Eq.~\ref{eq: mean_avg_rep_avg_rel_RMSE}), while the right panel shows the covariance network evaluation errors (Eq.~\ref{eq: cov_avg_rep_avg_rel_F}). Each point corresponds to a network trained with a specific \((M,N)\) configuration, and each colored curve represents a fixed total computational budget $B$ (Eq.~\ref{eq: fixed_budget}), with the horizontal axis indicating how that budget is allocated across \(N\) parameter settings.}
    \label{fig:budget_curves}
\end{figure*}

\begin{figure*}[ht!]
    \centering
    \includegraphics[width=\textwidth]{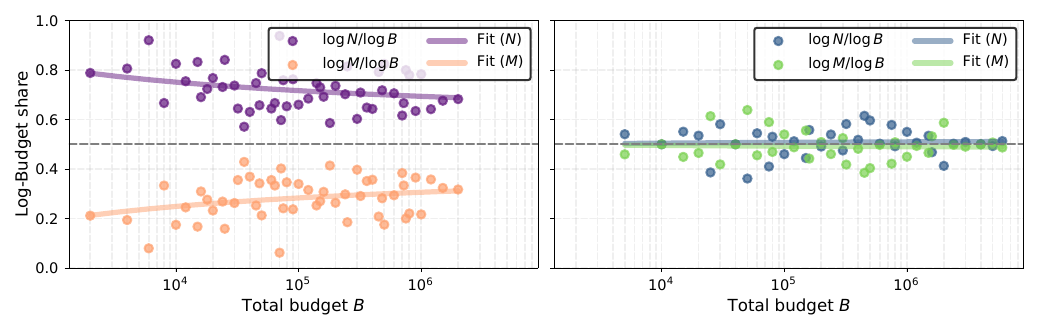}
    \caption{Optimal log-budget allocations for training data construction under a fixed computational budget $B$ (Eq.~\ref{eq: fixed_budget}). Each point shows the allocation $(M,N)$ in that minimizes the error (Eq.~\ref{eq: mean_avg_rep_avg_rel_RMSE} for the mean networks and Eq.~\ref{eq: cov_avg_rep_avg_rel_F} for the covariance networks) at a fixed budget $B$, expressed as the shares $\log N / \log B$ and $\log M / \log B$. Solid lines show fitted trends under the fixed-budget constraint. The dashed line at $0.5$ indicates a balanced allocation, $M \sim N \sim B^{1/2}$.}
    \label{fig:optimal_allocations}
\end{figure*}

We now investigate how the allocation of simulation budget between the number of parameter settings, \(N\), and the number of \ac{MC} replicates per parameter, \(M\), affects performance at fixed total budget $B$.

At a high level, this allocation induces a trade-off between coverage of the parameter space (increasing \(N\)) and reduction of \ac{MC} noise (increasing \(M\)). While the previous section showed that total budget determines the overall scale of achievable error, here we show that the way this budget is allocated plays a critical role in determining performance.

In Figure~\ref{fig:budget_curves}, we examine a range of representative fixed simulation budgets $B$ that contain multiple \((M,N)\) configurations used for model training. Each colored curve corresponds to a representative fixed budget \(B\), and shows the the evaluation errors (for mean networks Eq.~\ref{eq: mean_avg_rep_avg_rel_RMSE} and for covariance networks Eq.~\ref{eq: cov_avg_rep_avg_rel_F}) across different allocations of \(M\) and \(N\) satisfying this constraint. Because \(B\) is fixed along each curve, increasing \(N\) necessarily decreases \(M\), and vice versa. Each curve therefore represents a trade-off between allocating compute toward exploring more parameter settings (impact of $N$) or toward reducing \ac{MC} noise (impact of $M$) at each parameter.

Across all budgets, we observe that error is not monotonic in \(N\). Instead, the curves exhibit a clear U-shaped behavior, indicating that both extremes (too few parameter settings or too few \ac{MC} replicates per parameter) lead to suboptimal performance. Allocating too much budget to \(N\) results in insufficient sampling per parameter and high \ac{MC} noise, while allocating too much to \(M\) limits coverage of the parameter space and reduces the network’s ability to generalize.

This effect is particularly pronounced for the covariance model, where the error curves exhibit a clear minimum at intermediate values of \(N\). This indicates that neither allocating the budget primarily toward parameter coverage nor toward replication alone is sufficient; instead, achieving low error requires a balance between the two. This behavior is consistent with the analysis in Section~\ref{methods: effect_on_noise}, where we showed that Monte Carlo noise affects covariance estimation through both variance and bias induced by the nonlinear transformation of empirical estimates. As a result, increasing the number of replicates per parameter (\(M\)) plays a more critical role in stabilizing the training targets, while sufficient parameter coverage remains necessary for generalization. The observed interior optimum therefore reflects a stronger bias variance trade-off between \(M\) and \(N\) for the covariance model.

For the mean, the same qualitative U-shaped behavior is present, but with a notably different structure from the covariance case. In particular, as the total budget increases, the location of the minimum consistently shifts toward larger values of \(N\), indicating that allocating additional compute toward parameter coverage yields the greatest improvement in performance. Importantly, the range of Monte Carlo replicates considered for the mean model is relatively small (\(M \in [25, 200]\)), especially compared to the covariance setting. Despite this restriction, the error curves still exhibit clear interior minima. This indicates that the optimal number of replicates lies within this small range of \(M\), and that larger values of \(M\) are not required to achieve minimal error. If additional replication were necessary, we would instead expect the minimum to occur at the boundary corresponding to the largest available \(M\).
Together with the pronounced rightward shift of the curves in \(N\), these results show that mean estimation is primarily driven by parameter-space coverage once a modest level of replication is reached. This behavior is consistent with the analysis in Section~\ref{methods: effect_on_noise}, where Monte Carlo noise in the mean mapping was shown to act purely as variance, so that increasing \(M\) beyond a certain point yields diminishing returns while increasing \(N\) continues to improve generalization.

To further quantify this effect, we extract, at each fixed budget $B$, the \((M,N)\) configuration that minimizes error (restricted to budgets with at least three configurations) and analyze how the optimal allocation scales with \(B\). Figure~\ref{fig:optimal_allocations} shows the corresponding log-budget shares of the extracted optimum points, defined as \(\log N / \log B\) and \(\log M / \log B\). A value of \(0.5\) corresponds to a balanced allocation, where \(M \sim N \sim B^{1/2}\).

Across a wide range of budgets, the optimal allocations cluster around approximately constant log-share values, suggesting a power-law relationship of the form
\[
M(B) \propto B^{\alpha_M}, \qquad N(B) \propto B^{\alpha_N}, \qquad \alpha_M + \alpha_N = 1.
\]
Because the total budget satisfies \(B = MN\), the prefactors in these relationships must also be consistent with this constraint. In particular, writing \(M(B) \approx c_M B^{\alpha_M}\) and \(N(B) \approx c_N B^{\alpha_N}\), we require \(c_M c_N \approx 1\), ensuring that the product \(MN\) scales correctly with \(B\). For the covariance, fitting this model yields$
M(B) \approx 1.15\, B^{0.48}$ and $N(B) \approx 0.87\, B^{0.52}$.

The exponents are very close to the balanced allocation \((\alpha_M, \alpha_N) = (0.5, 0.5)\), with only minor deviation. This indicates that both \ac{MC} replication and parameter coverage contribute comparably to performance, consistent with the need to balance variance reduction and generalization observed in the error curves.

For the mean, the fitted scaling exhibits a more pronounced imbalance of $M(B) \approx 0.21\, B^{0.42}$ and $N(B) \approx 4.86\, B^{0.58}$. While the exponents still sum to one, the deviation from the balanced case is larger, and the prefactors differ substantially. In particular, the small prefactor on \(M\) indicates that only a relatively small number of \ac{MC} replicates per parameter is required before diminishing returns are reached, whereas the large prefactor on \(N\) reflects the continued benefit of increasing parameter coverage. Together, these results show that the optimal allocation is not only asymptotically biased toward \(N\), but also allocates relatively little budget to replication across the range of budgets considered. This behavior is consistent with the earlier observation that mean estimation converges quickly and becomes relatively insensitive to further reductions in \ac{MC} noise.

Taken together, these results show that optimal simulation design depends not only on total computational budget, but also on how that budget is allocated. In particular, covariance estimation requires a balanced allocation between replication and parameter coverage to manage its higher sensitivity to \ac{MC} noise, whereas mean estimation quickly becomes coverage-driven, benefiting primarily from broader exploration of the parameter space.

\subsection{Amortization of simulation cost}

Now that we understand how to allocate simulation budget to train neural networks with minimal error under a fixed training cost $B$, we examine how this learned model can amortize computational cost relative to brute-force \ac{MC} estimation.

We restrict attention to $(M,N)$ configurations corresponding to near-optimal allocation of simulation budget. For each such configuration, the neural network is trained once. Thus the computational cost, $B$, is incurred entirely upfront during training.

After training, the neural network defines a global mapping from parameters to predicted moments, and can be evaluated at new parameter values at negligible additional cost. As a result, its predictive accuracy on a test set does not depend on the number of test points $K$, and the total computational cost remains fixed at the training budget $B$.

In contrast, brute-force \ac{MC} estimation must be performed independently at each parameter value. To estimate moments at $K$ target parameters with $M$ \ac{MC} replicates per parameter requires total cost $ M \times K$. Consequently, the computational cost of brute-force estimation scales linearly with the number of evaluation points.

\begin{figure*}[ht!]
    \centering
    \includegraphics[width=\textwidth]{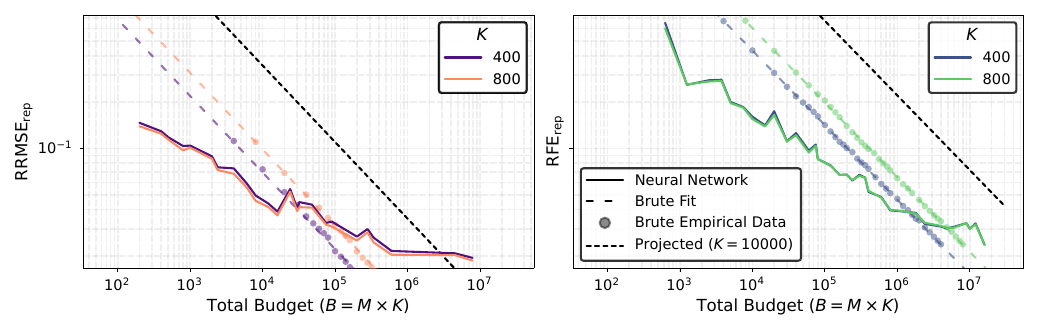}
    \caption{Learning amortization of moment estimation under increasing evaluation cost. Each color corresponds to a different number of evaluation points $K$. Solid lines show the neural surrogate error, evaluated over $K$ test points after training with optimally allocated budgets $(M,N)$ (total simulation cost $B$ (Eq.~\ref{eq: fixed_budget})), where evaluation incurs negligible additional cost. Dashed lines show the corresponding error from direct \ac{MC} simulation at $K$ evaluation points, with fitted trends. Dotted lines indicate the projected \ac{MC} cost when scaling to $K = 10^4$ evaluation points. The left panel reports mean trajectory error, while the right panel reports covariance error.}
    \label{fig:amortization}
\end{figure*}

Figure~\ref{fig:amortization} shows test error as a function of total compute for both approaches, for varying numbers of target parameters $K \in \{400, 800\}$. The error on the vertical axis is evaluated over a test set of $K$ parameter values.

For the neural network, each curve corresponds to a model trained at a given budget $B$, and remains unchanged as $K$ varies, since evaluation requires no additional simulation and predictive accuracy is stable across test set sizes.

For brute-force \ac{MC}, each curve corresponds to a fixed number of target parameters $K$. Moving along a curve reflects increasing the number of \ac{MC} replicates $M$ per parameter, which increases the total budget $M \times K$ and reduces error. The horizontal shift between curves is driven entirely by $K$. Meaning that increasing the number of target parameters increases the total compute proportionally, shifting the curves to the right due to the need to recompute simulations independently at each parameter.

At small $K$, the two approaches can achieve comparable accuracy, particularly for the mean. However, as $K$ increases, brute-force estimation requires proportionally more compute to achieve the same accuracy, while the neural network reuses its fixed training cost across all evaluation points.

In principle, brute-force \ac{MC} estimation is consistent: as the number of replicates $M \to \infty$, it converges to the true population quantities. However, this asymptotic regime is not only impractical due to finite computational resources, but becomes fundamentally unattainable as the number of evaluation points $K$ grows. Since brute-force estimation must be performed independently at each parameter, achieving high accuracy across many parameter values would require computational cost that scales as $M \times K$, which quickly becomes prohibitively large.

By contrast, the neural network approach is designed for this fixed-budget regime. While its error may plateau as the training budget $B$ increases, it can achieve sufficiently low error within a given computational budget and then reuse this accuracy across an arbitrarily large number of evaluations. As a result, when dense evaluation over the parameter space is required, the neural network provides a practical advantage by delivering stable accuracy at fixed cost, whereas brute-force methods must trade off accuracy against the number of evaluation points.

Taken together, these results demonstrate a clear amortization effect: the upfront cost of generating training data and fitting the neural network can be spread across many downstream evaluations, leading to substantial computational savings as the number of target parameters increases.

\subsection{Learned moments in downstream tasks}
We finally assess whether the learned mean and covariance are jointly accurate when used in a downstream transformation. While previous sections evaluate these quantities individually, many applications require them to be used together to capture and manipulate the dependence structure of stochastic trajectories.

Whitening provides a direct and interpretable diagnostic of this joint accuracy. To evaluate this, we apply the whitening transformation
$z = \Sigma_\theta^{-1/2}(x - \mu_\theta)$,
where $x$ denotes a trajectory sampled from the CTMC at parameter $\theta$, and $\mu_\theta$ and $\Sigma_\theta$ denote the learned mean and covariance.

If the mean and covariance are correctly estimated, the transformed trajectories should have identity covariance, implying no correlation across time points. Deviations from this behavior therefore directly reflect errors in the learned moments.

In our setting, each component corresponds to the infected population at a fixed time point, so whitening removes temporal dependence across the trajectory. This yields a clear visual criterion: accurate moment estimates produce an empirical correlation matrix close to the identity, while residual structure indicates inaccuracies in the learned moments.

We fix a parameter setting $(\alpha, \beta) = (0.55, 0.55)$ and generate 300 independent trajectories from the \ac{CTMC} model. Each trajectory consists of the infected population evaluated at discrete time points $t = 1, \dots, 13$, and is therefore represented as a vector in $\mathbb{R}^{13}$. Temporal dependence in the stochastic process manifests as correlation across these time dimensions.

We apply the whitening transformation using the learned mean $\mu_\theta$ and covariance $\Sigma_\theta$, obtained from neural networks trained under a fixed simulation budget (with $M = 15, N = 6000$ for the mean model and $M = 200, N = 400$ for the covariance model). All plots in Figure~\ref{fig:whitening} are generated from these trajectories and learned moments.

We compute empirical correlation matrices across time points, where each entry reflects correlation between the infected population at two time points. Figure~\ref{fig:whitening} shows these correlation matrices before and after whitening. The raw trajectories exhibit strong temporal correlations, reflecting the underlying epidemic dynamics. After applying the whitening transformation using the learned moments, this structure is largely removed, and the resulting correlation matrix is close to diagonal.

\begin{figure*}[ht!]
    \centering
    \includegraphics[width=\textwidth]{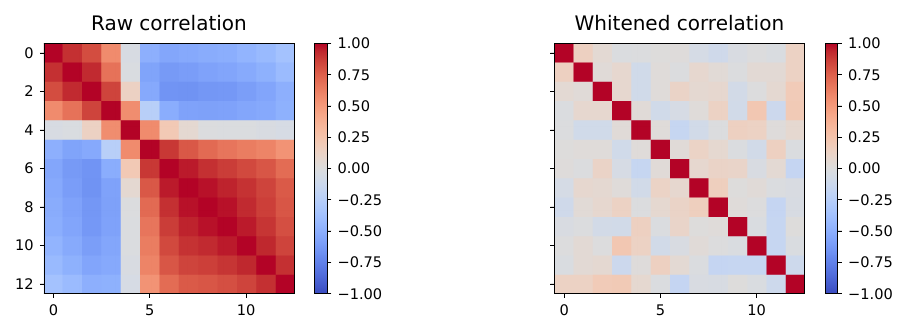}
    \caption{Empirical correlation matrices of \acf{CTMC} trajectories before and after whitening using learned moments. Each matrix shows correlation across time points for 300 simulated trajectories at $(\alpha, \beta) = (0.55, 0.55)$. The left panel shows strong temporal dependence in the raw trajectories. The right panel shows that applying the whitening transformation using the learned mean and covariance largely removes this dependence, resulting in an approximately identity correlation structure.}
    \label{fig:whitening}
\end{figure*}
While some residual correlations remain, they are small relative to the original dependence structure (maximum absolute correlation approximately $0.22$). This indicates that the learned covariance captures the dominant temporal dependencies in the system, and that the learned mean properly centers the trajectories.

Importantly, whitening removes second-order (linear) dependence but does not imply independence unless the data are Gaussian. In this setting, the underlying process is not Gaussian, and the transformed trajectories are not expected to be independent. Rather, the goal is to remove second-order dependence, which is achieved to a high degree of accuracy.

These results demonstrate that the learned mean and covariance are not only accurate individually, but also jointly consistent in a way that enables effective decorrelation of stochastic trajectories. This provides a functional validation of the learned moments, showing that they can be used together in downstream tasks that rely on removing dependence structure or standardizing model outputs.

\section{Discussion}\label{discussion}
In this work, we developed a simulation-based framework for learning parameter-to-moment mappings in \ac{CTMC} models where moments are not available in closed form. By learning these mappings from Monte Carlo outputs, we provide a more efficient approach to estimating moments across the parameter space under a fixed computational budget, rather than repeatedly simulating at each parameter value.

A central insight is that, under our chosen network architectures, Monte Carlo noise affects the mean and covariance differently, leading to different allocation strategies. Mean estimation benefits more from increased parameter-space coverage, as Monte Carlo noise primarily contributes variance without altering the underlying regression function. In contrast, covariance estimation requires additional replication to control both variance and bias introduced through nonlinear transformations. This highlights that simulation effort should be tailored to the specific moment being learned rather than applied uniformly.

We demonstrated these findings on a nontrivial stochastic system, specifically the \ac{SIR} model, focusing on a single model to study performance in depth across multiple aspects, including predictive accuracy, computational efficiency, and budget allocation. This focused analysis provides a clearer view of the underlying tradeoffs, though extending these results to additional models remains an important direction for future work.
Future work will explore how these strategies scale to more complex \ac{CTMC} models, including higher-dimensional parameter spaces and longer time horizons. Additional directions include using more expressive neural network architectures, incorporating the initial condition as an input, and developing adaptive strategies to allocate simulation effort or stop training based on downstream performance.
\backmatter

\bmhead{Supplementary information}

This article includes supplementary material provided as a separate document.

\bmhead{Acknowledgments}

This material is based upon work supported by the National Science Foundation under Grant No. DMS-2045843. MP and OPF were supported by this grant during the preparation of this work.

\section*{Declarations}

\bmhead{Funding}
This material is based upon work supported by the National Science Foundation under Grant No. DMS-2045843. MP and OPF were supported by this grant during the preparation of this work.

\bmhead{Competing interests}
The authors declare no competing interests.

\bmhead{Ethics approval and consent to participate}
Not applicable.

\bmhead{Consent for publication}
Not applicable.

\bmhead{Data availability}
All data are generated via simulation and can be reproduced using the code provided by the authors.

\bmhead{Materials availability} 
Not applicable.

\bmhead{Code availability}
All code used for simulation, analysis, and plotting results is available at \url{https://github.com/madison-git/CTMC_Moment_Learners}.

\bmhead{Author contributions}
MP conceived the study, performed the analysis, developed the methodology, implemented the code, and wrote the manuscript. OPF supervised the research and provided guidance and feedback on the manuscript.

\begin{appendices}

\section{Unbiasedness of Estimators}
\label{appendix:bruteforce_consistency}

We treat $\theta$ as fixed and suppress dependence on $\theta$ for clarity.

\subsection{Unbiasedness of Mean}

\label{app:mean_unbiased_app}
Since $\{\mathbf{Y}^{(m)}\}_{m=1}^M$ are i.i.d. with $\mathbb{E}[\mathbf{Y}^{(m)}] = \mu$, we have
\begin{equation*}
\begin{aligned}
\mathbb{E}[\hat{\mu}_M]
&= \mathbb{E}\left[\frac{1}{M}\sum_{m=1}^M \mathbf{Y}^{(m)}\right] \\
&= \frac{1}{M}\sum_{m=1}^M \mathbb{E}[\mathbf{Y}^{(m)}] \\
&= \mu.
\end{aligned}
\end{equation*}

\subsection{Unbiasedness of Covariance}
\label{app:cov_unbiased_app}
Define the centered variables $\mathbf{Z}^{(m)} := \mathbf{Y}^{(m)} - \mu$, so that $\mathbb{E}[\mathbf{Z}^{(m)}]=0$,
 $\mathrm{Cov}(\mathbf{Y}) = \mathbb{E}[\mathbf{Z}^{(m)}(\mathbf{Z}^{(m)})^\top] = \Sigma$.
We write
\begin{equation*}
\mathbf{Y}^{(m)} - \hat{\mu}_M
    =
\mathbf{Z}^{(m)} - \bar{\mathbf{Z}}, \quad
\bar{\mathbf{Z}} := \frac{1}{M}\sum_{k=1}^M \mathbf{Z}^{(k)}.
\end{equation*}
Thus,
\begin{equation*}
\hat{\Sigma}_M
=
\frac{1}{M-1}\sum_{m=1}^M
(\mathbf{Z}^{(m)} - \bar{\mathbf{Z}})
(\mathbf{Z}^{(m)} - \bar{\mathbf{Z}})^\top.
\end{equation*}
Expanding,
\begin{equation*}
\begin{aligned}
&\sum_{m=1}^M
(\mathbf{Z}^{(m)} - \bar{\mathbf{Z}})
(\mathbf{Z}^{(m)} - \bar{\mathbf{Z}})^\top\\
&\qquad = \sum_{m=1}^M \mathbf{Z}^{(m)}(\mathbf{Z}^{(m)})^\top   - M\,\bar{\mathbf{Z}}\bar{\mathbf{Z}}^\top.
\end{aligned}
\end{equation*}
Taking expectations,
\begin{equation*}
   \mathbb{E}\left[\sum_{m=1}^M \mathbf{Z}^{(m)}(\mathbf{Z}^{(m)})^\top\right]
=
M \Sigma, 
\end{equation*}
and
\begin{equation*}
\mathbb{E}\left[\bar{\mathbf{Z}}\bar{\mathbf{Z}}^\top\right]
=
\frac{1}{M}\Sigma,
\end{equation*}
since $\bar{\mathbf{Z}}$ has covariance $\Sigma/M$.
Therefore,
\begin{equation*}
\begin{aligned}
&\mathbb{E}\Bigg[
\sum_{m=1}^M
(\mathbf{Z}^{(m)} - \bar{\mathbf{Z}})
(\mathbf{Z}^{(m)} - \bar{\mathbf{Z}})^\top
\Bigg]\\
&\qquad =M\Sigma - M \cdot \frac{1}{M}\Sigma \\
&\qquad = (M-1)\Sigma.
\end{aligned}
\end{equation*}
Dividing by $M-1$ gives
\begin{equation*}
    \mathbb{E}[\hat{\Sigma}_M] = \Sigma.
\end{equation*}

\section{Structural properties of covariance estimation and parameterization}
\label{app:covariance}

This appendix justifies the preprocessing and parameterization steps used in the covariance model, including symmetrization, diagonal loading, and the Cholesky representation.

\subsection{Symmetrization of empirical covariance matrices}
\label{app:sym}
In exact arithmetic, the empirical covariance $\hat \Sigma_M (\theta)$, defined in \eqref{eq:mc_cov}, is symmetric since it is a sum of symmetric outer products.

However, in numerical representations, small asymmetries may arise due to floating-point rounding and accumulation order. These asymmetries do not reflect statistical variability, but are purely numerical artifacts.

Any matrix $A$ can be decomposed as
\[
A = \frac{1}{2}(A + A^\top) + \frac{1}{2}(A - A^\top),
\]
where the first term is symmetric and the second is antisymmetric. Since covariance matrices must be symmetric by definition, we enforce this constraint by replacing $\hat \Sigma_M(\theta)$ with its symmetric part:
\[
\hat \Sigma_M^{\mathrm{sym}}(\theta) = \frac{1}{2}(\hat \Sigma_M (\theta) + \hat \Sigma_M^\top(\theta)).
\]

This operation removes only the antisymmetric numerical artifact and leaves the symmetric component unchanged. In exact arithmetic, this step has no effect.

\subsection{Diagonal loading and positive definiteness}
\label{app:loading}
While empirical covariance matrices are positive semidefinite in theory, they may fail to be strictly positive definite in practice due to finite sample size or numerical error. This can prevent stable Cholesky factorization.
To address this, we apply diagonal loading:
\begin{equation*}
    \hat \Sigma_M \leftarrow \hat \Sigma_M + \lambda I,
    \quad \lambda > 0.
\end{equation*}
This shifts all eigenvalues of $\hat \Sigma_M$ by $\lambda$. If $\lambda$ exceeds the magnitude of any negative or near-zero eigenvalues, the resulting matrix becomes strictly positive definite. In practice, $\lambda$ is increased adaptively until the Cholesky factorization succeeds. This avoids explicit eigenvalue computation and ensures numerical stability while introducing only a small perturbation to the original estimate.

\section{The Effect of MC Noise on NN Targets}
\label{app:effect_on_mc_noise}

\subsubsection{Effect of Monte Carlo noise on the mean learner}
\label{app:mc_effect_mean_learner}
The training targets $\hat{\mu}_M(\theta)$ defined in \eqref{eq:mc_mean} are Monte Carlo estimates of the true mean $\mu(\theta)$ defined in \eqref{eq:pop_mean}. As a result, the learning problem corresponds to regression with noisy observations of the underlying function $\mu(\theta)$.

Let $\tilde{\theta} \sim \mathrm{Unif}(\Theta)$, and let $\mathbf{Y}(\tilde{\theta} )$ denote the corresponding trajectory vector defined in Section~\ref{trajectory_extraction}. We consider the population risk associated with training on Monte Carlo targets,
\begin{equation}
R_M(f)
=
\mathbb{E}\big[
\| f(\tilde{\theta} ) - \hat{\mu}_M(\tilde{\theta} ) \|_2^2
\big],
\label{eq:population_risk}
\end{equation}

\begin{proposition}
\label{prop:mc_decomposition_main}
Assume $\mathbb{E}\|\mathbf{Y}(\tilde{\theta} )\|_2^2 < \infty$. Then for any measurable function $f$,
\[
R_M(f)
=
\mathbb{E}\big[ \| f(\tilde{\theta} ) - \mu(\tilde{\theta} ) \|_2^2 \big]
+
\frac{1}{M}\, \mathbb{E}\big[ \mathrm{tr}\, \Sigma(\tilde{\theta} ) \big],
\]
where $\mu(\theta)$ and $\Sigma(\theta)$ are defined in \eqref{eq:pop_mean} and \eqref{eq:pop_cov}, respectively. See Appendix~\ref{app:pop_risk_min} for details.
\end{proposition}

\begin{corollary}
\label{cor:mc_target_main}
The population risk minimizer of $R_M$ over all square-integrable functions is
\[
f^\star(\theta) = \mu(\theta).
\]
\end{corollary}

Proposition~\ref{prop:mc_decomposition_main} shows that the use of Monte Carlo estimators introduces an additive variance term that scales as $1/M$, but does not alter the underlying regression function. Consequently, even though the network is trained on noisy targets, the optimal predictor remains the true mean mapping $\theta \mapsto \mu(\theta)$.

In practice, the network is trained using the empirical risk defined in \eqref{eq:mean_loss}, which approximates the population risk \eqref{eq:population_risk} using a finite set of parameter samples $\{\theta_i\}_{i=1}^N$. Proposition~\ref{prop:erm_bound_main} below shows that the generalization error is controlled by the deviation between empirical and population risk.

\begin{proposition}
\label{prop:erm_bound_main}
Let $\hat{f}_{N,M} \in \arg\min_{f \in \mathcal{F}} \hat{R}_{N,M}(f)$ and
\[
f^\star_{\mathcal{F}} \in \arg\min_{f \in \mathcal{F}} R_M(f).
\]
Then
\[
R_M(\hat{f}_{N,M}) - R_M(f^\star_{\mathcal{F}})
\le
2 \sup_{f \in \mathcal{F}} \big| R_M(f) - \hat{R}_{N,M}(f) \big|.
\]
\end{proposition}

The proposition above follows from the standard ERM analysis under uniform convergence, see \cite[Section 2.3]{shalev2014understanding} but is also proved for our given setting in Appendix~\ref{app:erm_bound_app} for convenience.

Together, these results imply that the mean network $f_\mu$, trained via \eqref{eq:mean_loss}, is learning the true moment mapping $\theta \mapsto \mu(\theta)$ defined in \eqref{eq:pop_mean}, with $M$ primarily controlling the variance of the training targets and $N$ governing generalization performance.

\subsubsection{Effect of Monte Carlo noise on the covariance learner}
\label{app:mc_effect_cov_learner}
For the covariance model, the training target is obtained by applying a nonlinear transformation to the empirical covariance estimator. In particular, the training target is the vectorized lower diagonal entries of the Cholesky factor (\ref{eq:vec_cholesky}). As a result, even though $\hat{\Sigma}_M(\theta)$ is an unbiased estimator of $\Sigma(\theta)$, the transformed target $\hat{\ell}_M(\theta)$ is generally not an unbiased estimator of $\ell(\theta)$. Thus, the covariance learner is trained on targets whose bias depends on $M$. To view the implications of this, we look at how the empirical risk (\ref{eq:emp_risk_cov}) relates to the population risk. Let $\tilde{\theta} \sim \mathrm{Unif}(\Theta)$, and define the population risk
\begin{equation}
R_M^{\Sigma}(f)
=
\mathbb{E}\Big[
\| f(\tilde{\theta}) - \hat{\ell}_M(\tilde{\theta}) \|_2^2
\Big].
\label{eq:population_risk_cov}
\end{equation}

\begin{proposition}
\label{prop:cov_target_main}
Let $g_M(\theta) = \mathbb{E}\big[\hat{\ell}_M(\theta)\big]$. Assume $\mathbb{E}\|\hat{\ell}_M(\tilde{\theta})\|_2^2 < \infty$. Then for any measurable function $f$,
\begin{equation*}
R_M^{\Sigma}(f)
=
\mathbb{E}\big[\|f(\tilde{\theta}) - g_M(\tilde{\theta})\|_2^2\big]
+
\mathbb{E}\big[\|\hat{\ell}_M(\tilde{\theta}) - g_M(\tilde{\theta})\|_2^2\big].
\end{equation*}
In particular, the population risk minimizer of $R_M^{\Sigma}$ over all square-integrable functions is
\[
f_M^\star(\theta) = g_M(\theta).
\]
Details of this proof can be seen in Appendix~\ref{app:cov_pop_min_app}.
\end{proposition}

Proposition~\ref{prop:cov_target_main} shows that, unlike the mean learner, the covariance learner is optimized toward the finite-M regression function $g_M(\theta)$, rather than the population target $\ell(\theta)$. Thus, Monte Carlo noise does not enter the risk merely as an additive constant, but instead alters the effective target function through the nonlinear transformation.

We next show that this finite-$M$ bias vanishes asymptotically. Intuitively, as the number of Monte Carlo samples increases, the empirical covariance converges to the population covariance, and the corresponding Cholesky factor converges to the population Cholesky factor.

\begin{proposition}
\label{prop:cov_consistency_main}
Assume that $\Sigma(\theta)$ is positive definite for each $\theta \in \Theta$. 
Further assume that the sample covariance estimator $\hat{\Sigma}_M(\theta)$ converges almost surely to $\Sigma(\theta)$ as $M \to \infty$ and that the preprocessing steps are asymptotically negligible. Then $\hat{\ell}_M(\theta) \to \ell(\theta)$ almost surely as $M \to \infty$.
If, in addition, the family $\{\hat{\ell}_M(\theta)\}_{M \ge 1}$ is uniformly integrable, then
$g_M(\theta) = \mathbb{E}[\hat{\ell}_M(\theta)] \to \ell(\theta)$
as $M \to \infty$.
\end{proposition}
Proposition~\ref{prop:cov_consistency_main} shows that although the covariance learner is trained on biased finite-$M$ targets, this bias disappears asymptotically under suitable regularity conditions. In particular, the effective regression function $g_M(\theta)$ converges to the population Cholesky target $\ell(\theta)$ as $M \to \infty$.

In practice, the covariance network is trained using the empirical risk (\ref{eq:emp_risk_cov}), which approximates the population risk \eqref{eq:population_risk_cov} using finitely many parameter samples $N$. The covariance learner behaves differently from the mean learner in two important ways. First, finite Monte Carlo sampling affects not only the variance of the training targets, but also the regression function itself through the nonlinear Cholesky transformation. Second, increasing $M$ improves covariance learning in two ways: it reduces the variability of the empirical covariance estimator and decreases the bias between the finite-$M$ target $g_M(\theta)$ and the population Cholesky target $\ell(\theta)$. This suggests that larger values of $M$ are more critical for the covariance learner than for the mean learner.

\section{Theoretical analysis for the mean learner}

\subsection{Population risk decomposition}
\label{app:pop_risk_min}
Below is a proof of Proposition~\ref{prop:mc_decomposition_main}

\begin{proof}
Using the definition of the Monte Carlo estimator in \eqref{eq:mc_mean}, write
\begin{equation*}
   f(\tilde{\theta}) - \hat{\mu}_M(\tilde{\theta})
 =
\big(f(\tilde{\theta}) - \mu(\tilde{\theta})\big)
+
\big(\mu(\tilde{\theta}) - \hat{\mu}_M(\tilde{\theta})\big).    
\end{equation*}
Then
\begin{equation*}
    \begin{aligned}
    \|f(\tilde{\theta}) - \hat{\mu}_M(\tilde{\theta})\|_2^2
&=
\|f(\tilde{\theta}) - \mu(\tilde{\theta})\|_2^2
+
\|\hat{\mu}_M(\tilde{\theta}) - \mu(\tilde{\theta})\|_2^2\\
&\qquad +
2 \left\langle f(\tilde{\theta}) - \mu(\tilde{\theta}), \mu(\tilde{\theta}) - \hat{\mu}_M(\tilde{\theta}) \right\rangle.
    \end{aligned}
\end{equation*}
Taking expectations, the cross term vanishes since
\begin{equation*}           \mathbb{E}\big[\hat{\mu}_M(\tilde{\theta}) \mid \tilde{\theta}\big] = \mu(\tilde{\theta}),
\end{equation*}
by unbiasedness of the estimator in \eqref{eq:mc_mean}.
For the variance term, since $\hat{\mu}_M(\tilde{\theta})$ is the average of $M$ independent samples of $\mathbf{Y}(\tilde{\theta})$,
\[
\mathbb{E}\|\hat{\mu}_M(\tilde{\theta}) - \mu(\tilde{\theta})\|_2^2
=
\frac{1}{M}\, \mathbb{E}\big[\mathrm{tr}\,\Sigma(\tilde{\theta})\big],
\]
where $\Sigma(\theta)$ is defined in \eqref{eq:pop_cov}.
Combining terms yields the result.
\end{proof}

\subsection{Finite-sample ERM bound}
\label{app:erm_bound_app}
Below is a proof of Proposition~\ref{prop:erm_bound_main}
\begin{proof}
Add and subtract empirical risks:
\begin{equation*}
\begin{aligned}
 R_M(\hat{f}_{N,M}) - R_M(f^\star_{\mathcal{F}})
&=
\big(R_M(\hat{f}_{N,M}) - \hat{R}_{N,M}(\hat{f}_{N,M})\big)
\\&+
\big(\hat{R}_{N,M}(\hat{f}_{N,M}) - \hat{R}_{N,M}(f^\star_{\mathcal{F}})\big)\\
&+ \big(\hat{R}_{N,M}(f^\star_{\mathcal{F}}) - R_M(f^\star_{\mathcal{F}})\big).
\end{aligned}
\end{equation*}
The middle term is non-positive since $\hat{f}_{N,M}$ minimizes $\hat{R}_{N,M}$. Thus,
\begin{equation*}
    \begin{aligned}
        R_M(\hat{f}_{N,M}) - R_M(f^\star_{\mathcal{F}})
\le
\big|R_M(\hat{f}_{N,M}) - \hat{R}_{N,M}(\hat{f}_{N,M})\big|\\
+
\big|\hat{R}_{N,M}(f^\star_{\mathcal{F}}) - R_M(f^\star_{\mathcal{F}})\big|.
    \end{aligned}
\end{equation*}
Each term is bounded by
\[
\sup_{f \in \mathcal{F}} \big| R_M(f) - \hat{R}_{N,M}(f) \big|,
\]
which gives the result.
\end{proof}

\section{Theoretical analysis for the covariance learner}

\subsection{Population Risk Decomposition and Minimizer}
\label{app:cov_pop_min_app}
Below is a proof of Proposition~\ref{prop:cov_target_main}. 

\begin{proof}
Let
\[
g_M(\theta) := \mathbb{E}[\hat{\ell}_M(\theta)].
\]
We write
\begin{equation*}
    \begin{aligned}
        f(\tilde{\theta}) - \hat{\ell}_M(\tilde{\theta})
        &=
        \bigl(f(\tilde{\theta}) - g_M(\tilde{\theta})\bigr)\\
        &\qquad \quad +
         \bigl(g_M(\tilde{\theta}) - \hat{\ell}_M(\tilde{\theta})\bigr).
    \end{aligned}
\end{equation*}
Expanding the squared norm gives
\begin{align*}
\|f(\tilde{\theta}) - \hat{\ell}_M(\tilde{\theta})\|_2^2
&=
\|f(\tilde{\theta}) - g_M(\tilde{\theta})\|_2^2
\\
&+
\|\hat{\ell}_M(\tilde{\theta}) - g_M(\tilde{\theta})\|_2^2\\ 
&
+
2\bigl(f(\tilde{\theta}) - g_M(\tilde{\theta})\bigr)^\top
\bigl(g_M(\tilde{\theta}) - \hat{\ell}_M(\tilde{\theta})\bigr).
\end{align*}
Taking expectations yields
\begin{align*}
R_M^{\Sigma}(f)
&=
\mathbb{E}\big[\|f(\tilde{\theta}) - g_M(\tilde{\theta})\|_2^2\big]\\
&+
\mathbb{E}\big[\|\hat{\ell}_M(\tilde{\theta}) - g_M(\tilde{\theta})\|_2^2\big] \\
&+
2\,\mathbb{E}\!\left[
\bigl(f(\tilde{\theta}) - g_M(\tilde{\theta})\bigr)^\top
\bigl(g_M(\tilde{\theta}) - \hat{\ell}_M(\tilde{\theta})\bigr)
\right].
\end{align*}
We now show that the cross term is zero. Conditioning on $\tilde{\theta}$,
\begin{equation*}
\begin{aligned}
\mathbb{E}\!&\left[
\bigl(f(\tilde{\theta}) - g_M(\tilde{\theta})\bigr)^\top
\bigl(g_M(\tilde{\theta}) - \hat{\ell}_M(\tilde{\theta})\bigr)
\right]\\
 &\qquad =
\mathbb{E}\!\left[
\bigl(f(\tilde{\theta}) - g_M(\tilde{\theta})\bigr)^\top
\mathbb{E}\!\left[g_M(\tilde{\theta}) - \hat{\ell}_M(\tilde{\theta}) \mid \tilde{\theta}\right]
\right].    
\end{aligned}
\end{equation*}
By definition of $g_M$,
\begin{equation*}
\mathbb{E}[\hat{\ell}_M(\tilde{\theta}) \mid \tilde{\theta}] = g_M(\tilde{\theta}),    
\end{equation*}
so the inner expectation is zero, and hence the cross term vanishes. Therefore,
\begin{equation*}
 R_M^{\Sigma}(f)
=
\mathbb{E}\big[\|f(\tilde{\theta}) - g_M(\tilde{\theta})\|_2^2\big]
+
\mathbb{E}\big[\|\hat{\ell}_M(\tilde{\theta}) - g_M(\tilde{\theta})\|_2^2\big].   
\end{equation*}
To identify the minimizer, observe that the second term does not depend on $f$. Thus minimizing $R_M^{\Sigma}(f)$ is equivalent to minimizing
\[
\mathbb{E}\big[\|f(\tilde{\theta}) - g_M(\tilde{\theta})\|_2^2\big],
\]
which is uniquely minimized by
\[
f_M^\star(\theta) = g_M(\theta).
\]
\end{proof}

\section{Moment Equations for the CTMC SIR Model}
\label{app:SIR_moment_derive}

We define the joint moments of the \ac{CTMC} \ac{SIR} model as
\begin{equation*}
M_{a,b}(t) = \sum_{s=0}^{N} \sum_{i=0}^{N} s^a i^b \, p_{s,i}(t),
\end{equation*}
where $p_{s,i}(t)$ is the probability of observing $s$ susceptible and $i$ infected individuals at time $t$.

\subsection{Derivation of the Moment Equation}

From the forward Kolmogorov equation,
\begin{equation*}
\begin{aligned}
\frac{dp_{s,i}(t)}{dt}
&= \beta (s+1)(i-1)p_{s+1,i-1}(t)\\
&\quad + \gamma (i+1)p_{s,i+1}(t)\\
&\quad - (\beta si + \gamma i)p_{s,i}(t),   
\end{aligned}
\end{equation*}

Multiplying both sides by $s^a i^b$ and summing over all $s,i$:
\begin{align*}
\frac{dM_{a,b}}{dt}
&= \sum_s \sum_i s^a i^b \frac{dp_{s,i}(t)}{dt}\\
&= \sum_s \sum_i s^a i^b
\beta (s+1)(i-1)p_{s+1,i-1}(t)\\
&\quad + \sum_s \sum_i s^a i^b\gamma (i+1)p_{s,i+1}(t)\\
&\quad- \sum_s \sum_i s^a i^b(\beta si + \gamma i)p_{s,i}(t).
\end{align*}
The first term represents infection. Letting $k = s+1$, $j = i-1$, and relabeling $(k,j)$ as $(s,i)$,
\begin{align*}
&\sum_s \sum_i s^a i^b (s+1)(i-1)p_{s+1,i-1}\\
&\qquad = \sum_s \sum_i (s-1)^a (i+1)^b si \, p_{s,i}(t)\\
&\qquad = \sum_{m=0}^{a} \sum_{n=0}^{b}
\binom{a}{m} \binom{b}{n}
(-1)^{a-m}
 \\
 &\qquad \qquad \times\sum_s \sum_i s^{m+1} i^{n+1} p_{s,i}(t)\\
&\qquad = \sum_{m=0}^{a} \sum_{n=0}^{b}
\binom{a}{m} \binom{b}{n}
(-1)^{a-m}
M_{m+1,n+1}(t).
\end{align*}
Multiplying by $\beta$ gives
\begin{equation*}
\beta \sum_{m=0}^{a} \sum_{n=0}^{b}
\binom{a}{m} \binom{b}{n}
(-1)^{a-m}
M_{m+1,n+1}(t).
\end{equation*}
The second term represents recovery. Letting $k = i+1$,
\begin{align*}
\sum_s \sum_i s^a i^b (i+1)p_{s,i+1}
&= \sum_s \sum_i s^a (i-1)^b i \, p_{s,i}(t)\\
&= \sum_{n=0}^{b}
\binom{b}{n} (-1)^{b-n}
M_{a,n+1}(t).
\end{align*}
Multiplying by $\gamma$ gives
\begin{equation*}
\gamma \sum_{n=0}^{b}
\binom{b}{n} (-1)^{b-n}
M_{a,n+1}(t).
\end{equation*}
The final term can be written directly in terms of $M$:
\begin{equation*}
\begin{aligned}
\sum_s \sum_i s^a i^b (\beta si + \gamma i)p_{s,i}(t)
&= \beta M_{a+1,b+1}(t) \\
&\quad+ \gamma M_{a,b+1}(t).
\end{aligned}
\end{equation*}
Thus, the moment equation for the SIR \ac{CTMC} is
\begin{align*}
\frac{dM_{a,b}}{dt}
&= \beta \sum_{m=0}^{a} \sum_{n=0}^{b}
\binom{a}{m} \binom{b}{n}
(-1)^{a-m}
M_{m+1,n+1}(t) \\
&\quad + \gamma \sum_{n=0}^{b}
\binom{b}{n} (-1)^{b-n}
M_{a,n+1}(t)\\
&\quad - \beta M_{a+1,b+1}(t)
- \gamma M_{a,b+1}(t).
\end{align*}

\subsection{Example: First Moment of the Infected Population}

The equation governing the first moment of $I$ is obtained by substituting $a = 0$, $b = 1$:

\begin{align*}
\frac{dM_{0,1}}{dt}
&= \beta (M_{1,1} + M_{1,2})
+ \gamma(-M_{0,1} + M_{0,2})\\
&\quad- \beta M_{1,2}
- \gamma M_{0,2}\\
&= \beta M_{1,1} - \gamma M_{0,1}.
\end{align*}
Thus,
\begin{equation*}
\frac{d}{dt}\mathbb{E}[I]
= \beta \mathbb{E}[SI] - \gamma \mathbb{E}[I].
\end{equation*}

\subsection{Example: Second-Order Moment}

Similarly, the equation governing the second-order moment $\mathbb{E}[SI] = M_{1,1}(t)$ is obtained by setting $a = 1$, $b = 1$:
\begin{align*}
\frac{dM_{1,1}}{dt}
&= \beta(-M_{1,1} - M_{1,2} + M_{2,1} + M_{2,2}) \\
&\quad + \gamma(-M_{1,1} + M_{1,2})
- \beta M_{2,2}
- \gamma M_{1,2}\\
&= \beta(-M_{1,1} - M_{1,2} + M_{2,1}) - \gamma M_{1,1}.
\end{align*}
Thus,
\begin{equation*}
\frac{d}{dt}\mathbb{E}[SI]
= \beta(-\mathbb{E}[SI] - \mathbb{E}[SI^2] + \mathbb{E}[S^2 I])
- \gamma \mathbb{E}[SI].
\end{equation*}

\subsection{Covariance}
\begin{equation*}
\mathrm{Cov}(S,I) = \mathbb{E}[SI] - \mathbb{E}[S]\mathbb{E}[I].
\end{equation*}
Taking the derivative,
\begin{equation*}
\frac{d}{dt}\mathrm{Cov}(S,I)
= \frac{d}{dt}\mathbb{E}[SI]
- \frac{d}{dt}(\mathbb{E}[S]\mathbb{E}[I]).
\end{equation*}
Since both $\frac{d}{dt}\mathbb{E}[S]$ and $\frac{d}{dt}\mathbb{E}[I]$ depend on $\mathbb{E}[SI]$, and $\frac{d}{dt}\mathbb{E}[SI]$ depends on higher-order moments, the covariance evolution also depends on higher-order moments.

\subsection{Lack of Closure}

The first moment depends on second-order moments, and second-order moments depend on third-order moments. This illustrates that the moment equations form an infinite hierarchy, where lower-order moments depend on higher-order moments, and thus are not closed without additional approximation.

\section{Additional SIR Experiment Details}
\label{app: add_exp_details}

\subsection{$M,N$ Configurations and Replicates}
\label{app: exp_configs}

For the mean model, we train neural networks using all pairings of the following parameter $N$ and replicate sizes $M$:

\begin{table}[ht!]
\centering
\caption{Mean Leaner $M$ and $N$ Training Configurations and Replicates per $N$}
\label{tab:mean_MN_config}
\begin{tabular}{|c|c|c|c|c|}
\hline
$M$ & $N$ & Rep per $N$ & $N$ cont. & Rep per $N$ cont.\\
\hline
2  & 100 & 6 & 15000 & 2\\
5  & 200 & 6 & 17000 & 2\\
10 & 400 & 5 & 20000 & 2\\
15 & 800 & 5 & 25000 & 2\\
20 & 1200 & 4 & 30000 & 2\\
25 & 1600 & 4 & 35000 & 1\\
40 & 2000 & 4 & 40000 & 1\\
50 & 3000 & 3 & 50000 & 1\\
75 & 4000 & 3 & - & -\\
90 & 5000 & 3 & - & -\\
100 & 6000 & 3 & - & -\\
125 & 7000 & 3 & - & -\\
150 & 8000 & 3 & - & -\\
175 & 9000 & 3 & - & -\\
200 & 12000 & 2 & - & -\\
\hline
\end{tabular}
\end{table}

For the covariance model, we train neural networks using all pairings of the following parameter $N$ and replicate sizes $M$:

\begin{table}[ht!]
\centering
\caption{Covariance Leaner $M$ and $N$ Training Configurations and Replicates per $N$}
\label{tab:cov_MN_config}
\begin{tabular}{|c|c|c|}
\hline
$M$ & $N$ & Rep per $N$ \\
\hline
25  & 25 & 6 \\
50  & 50 & 6 \\
75 & 100 & 6 \\
100 & 200 & 5 \\
200 & 400 & 4 \\
300 & 800 & 3 \\
500 & 1200 & 3\\
750 & 1600 & 3 \\
1000 & 2000 & 3\\
1500 & 2500 & 3 \\
2000 & 3000 & 3 \\
2500 & 3500 & 3\\
3000 & 4000 & 3\\
4000 & 4500 & 3\\
5000 & - & -\\
\hline
\end{tabular}
\end{table}

\subsection{Additional Plotting Metrics}
\label{app:additional_metrics}

We will further evaluate performance at a single test parameter by looking at the absolute error across time points given by

\begin{equation}
    e_t = |\hat \mu_t(\theta) - \mu_t(\theta)|,
    \label{eq: mean_entrywise_abs_error}
\end{equation}

as well as the entrywise relative error given by

\begin{equation}
    r_t = \frac{|\hat \mu_t(\theta) - \mu_t(\theta)|}{|y_t(\theta)| + \epsilon}.
    \label{eq: mean_entrywise_rel_error}
\end{equation}

For the covariance we look at relative entrywise error given by

\begin{equation}
    r_{i,j}(\theta) = \frac{|\hat \Sigma_{i,j}(\theta) - \Sigma_{i,j}(\theta)|}{|\Sigma_{i,j}(\theta)| + \epsilon}
    \label{eq:cov_entrywise_rel_error}
\end{equation}




\end{appendices}


\bibliography{sn-bibliography}

\end{document}


\title{Supplementary Material:
Learning Moment Maps for Continuous-Time Markov Chains under Monte Carlo Noise}

\author{Madison Pratt, Olivia Prosper-Feldman}

\maketitle

\section*{Supplementary Figures}

The figures below provide additional support for the claims presented in the Results section. Each subsection is titled to correspond with the relevant portion of the Results, indicating where the material is being supplemented

\begin{figure*}[t]
    \centering
    \includegraphics[width=\textwidth]{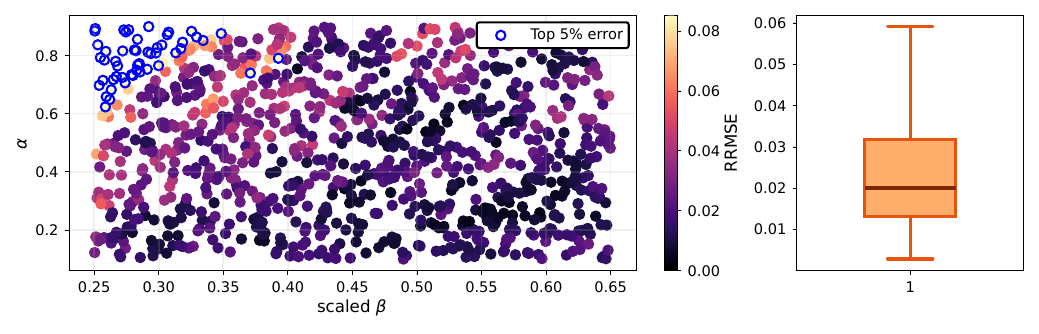}
    \caption{The left plot shows the $\text{RRMSE}$ per test parameter using the mean network trained with computational budget on the order of $10^5$. The colorbar is clipped to the 95th percentile $\text{RRMSE}$ values. The parameter values with $\text{RRMSE}$ higher than the 95th percentile are highlighted in blue. The right plot shows the distribution of the $\text{RRMSE}$.}  \label{sup_fig:mean_global_view}
\end{figure*}

\begin{figure*}[t]
    \centering
    \includegraphics[width=\textwidth]{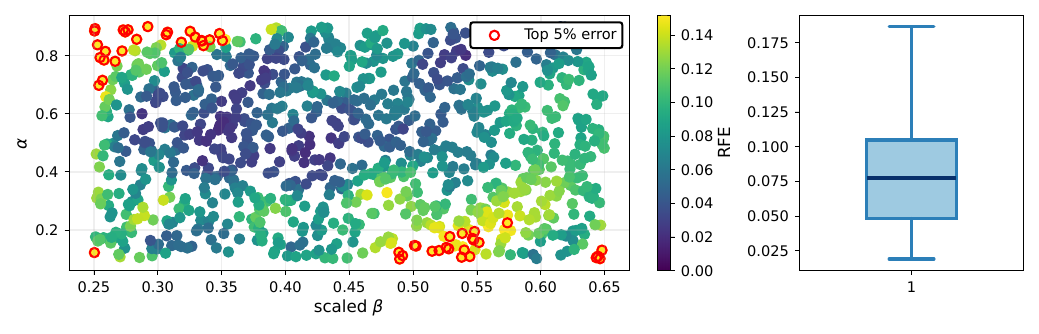}
    \caption{The left plot shows the $\text{RFE}$ per test parameter using the covariance network trained with computational budget on the order of $10^5$. The colorbar is clipped to the 95th percentile $\text{RFE}$ values. The parameter values with $\text{RFE}$ higher than the 95th percentile are highlighted in red. The right plot shows the distribution of the $\text{RFE}$.}  \label{sup_fig:cov_global_view}
\end{figure*}

\begin{figure*}[t]
    \centering
    \includegraphics[width=\textwidth]{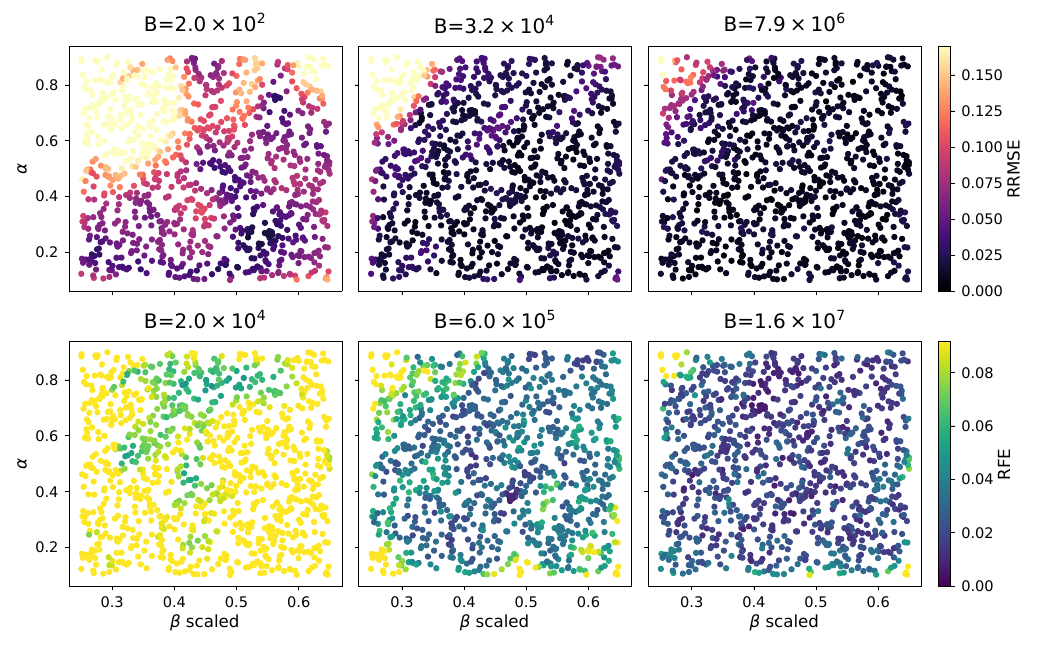}
    \caption{Top panel shows the RRMSE over the test parameter space over three different training budget sizes. Bottom panel shows the RFE over the test parameter space over three different budget sizes.}  \label{sup_fig:param_error_v_budget}
\end{figure*}